\documentclass[10pt,conference]{IEEEtran}
\IEEEoverridecommandlockouts
\usepackage{cite}
\usepackage{amsmath,amssymb,amsfonts}
\usepackage{amsthm}
\usepackage{amsbsy}
\usepackage{balance}
\usepackage{bm}
\usepackage{graphicx}
\usepackage{textcomp}
\usepackage{xcolor}
\usepackage{algorithm}
\usepackage{algpseudocode}
\algblock{Output}{EndOutput}
\algnotext{EndOutput}
\algnewcommand{\Initialize}[1]{%
	\State \textbf{Initialize:}
	\Statex \hspace*{\algorithmicindent}\parbox[t]{.8\linewidth}{\raggedright #1}
}
\usepackage{dblfloatfix}
\usepackage[outdir=./]{epstopdf}
\usepackage{enumitem}
\newlist{steps}{enumerate}{1}
\setlist[steps, 1]{label = Step \arabic*:}
\allowdisplaybreaks
\def\BibTeX{{\rm B\kern-.05em{\sc i\kern-.025em b}\kern-.08em
\usepackage{hyperref}
    T\kern-.1667em\lower.7ex\hbox{E}\kern-.125emX}}

\newtheorem{proposition}{Proposition} 
\theoremstyle{remark}

\newcommand{\tUR}{UR}
\newcommand{\tUT}{UT}
\newcommand{\UR}{{\mathrm{r}}}
\newcommand{\UT}{{\mathrm{t}}}
\newcommand{\MN}{{\mathrm{c}}}
\newcommand{\nr}{{n}}
\newcommand{\nrl}{n'}
\newcommand{\NR}{{N_\mathrm{r}}}

\newcommand{\NT}{{N_\mathrm{t}}}
\newcommand{\nm}{p}
\newcommand{\NM}{N}
\newcommand{\UTUR}{\mathrm{tr}} 
\newcommand{\URUT}{\mathrm{rt}}

\newcommand{\ymtest}{\Tilde{\mathbf{y}}_{\nm}}
\newcommand{\phiu}{\bm{\varphi}_{\nr}}
\newcommand{\phiuh}{\bm{\varphi}_{\nr}^{H}}
\newcommand{\phid}{\bm{\Phi}_{\UT}}
\newcommand{\phidh}{\bm{\Phi}_{\UT}^H}
\newcommand{\esp}{\mathbb{E}}
\newcommand{\bmest}{\mathbf{b}_{\nm}}
\newcommand{\omegam}{\Tilde{\bm{\omega}}_{\nm}}
\newcommand{\Frj}{\mathbf{F}_\UR^\mathrm{J}}
\newcommand{\Fmj}{\mathbf{F}^{\mathrm{J}}}

\begin{document}

\title{CSI Acquisition in Cell-Free Massive MIMO Surveillance Systems \vspace{-0.3em}

\author{

\IEEEauthorblockN{Isabella W. G. da Silva, Zahra Mobini, Hien Quoc Ngo, and Michail Matthaiou}

\IEEEauthorblockA{Centre for Wireless Innovation (CWI), Queen's University Belfast, U.K.}

\IEEEauthorblockA{E-mails: \{iwgdasilva01, zahra.mobini, hien.ngo, m.matthaiou\}@qub.ac.uk}

\vspace{-0.45cm}
}

\thanks{This work is a contribution by Project REASON, a UK Government funded project under the Future Open Networks Research Challenge (FONRC) sponsored by the Department of Science Innovation and Technology (DSIT). It was also supported by the U.K. Engineering and Physical Sciences Research Council (EPSRC) (grants No. EP/X04047X/1 and EP/X040569/1). The work of I. W. G. da Silva, Z.~Mobini, and  H.~Q.~Ngo was supported by the U.K. Research and Innovation Future Leaders Fellowships under Grant MR/X010635/1, and a research grant from the Department for the Economy Northern Ireland under the US-Ireland R\&D Partnership Programme. The work of M. Matthaiou was supported by the European Research Council (ERC) under the European Union’s Horizon 2020 research and innovation programme (grant agreement No. 101001331).}}

\maketitle

\begin{abstract}
We consider a cell-free massive multiple-input multiple-output (CF-mMIMO) surveillance system, in which multiple multi-antenna monitoring nodes (MNs) are deployed in either observing or jamming mode to disrupt the communication between a multi-antenna untrusted pair.
We propose a simple and effective channel state information (CSI) acquisition scheme at the MNs. Specifically, our approach leverages pilot signals in both the uplink and downlink phases of the untrusted link, coupled with minimum mean-squared error (MMSE) estimation. This enables the MNs to accurately estimate the effective channels to both the untrusted transmitter (UT) and untrusted receiver (UR), thereby yielding robust monitoring performance. We analyze the spectral efficiency (SE) performance of the untrusted links and of the monitoring
system, taking into account the proposed CSI acquisition
and successive MMSE cancellation schemes. The monitoring success probability (MSP) is then derived. Simulation results show that the CF-mMIMO surveillance system, relying on the proposed CSI acquisition scheme, can achieve monitoring performance close to that achieved by having perfect CSI knowledge of the untrusted link (theoretical upper bound), especially when the number of MNs is large.


\end{abstract}

\section{Introduction}\label{sec:Introduction}

As infrastructure-free wireless systems, such as device-to-device (D2D) and mobile ad-hoc communication paradigms advance, a novel topic within the space of physical-layer security (PLS), denoted as proactive monitoring, has garnered remarkable attention~\cite{8726325,Zahra:TIFS:2019}. 
Proactive monitoring involves legitimate nodes assuming the role of a monitor (or eavesdropper) to actively monitor unauthorized or malicious users who aim to exploit the wireless systems for illegal activities, cybercrime, or to jeopardize public safety~\cite{9187662,9673103}. 
In~\cite{9187662}, Feizi \textit{et al.} evaluated a  surveillance system comprised of one multi-antenna full-duplex (FD) MN, a pair of untrusted users, and a scheduled downlink user, where the transmit precoder and receive beamformer are optimized to maximize the eavesdropping non-outage probability. In~\cite{9673103}, Xu and Zhu studied a proactive monitoring scheme in which one legitimate FD MN observes multiple untrusted pairs and sends jamming or constructive signals to interfere with the communication between the untrusted links. 
In a recent study~\cite{paperzahra}, the CF-mMIMO infrastructure has been employed as a promising solution to enhance the monitoring capabilities within wireless surveillance frameworks. CF-mMIMO can offer high macro-diversity and ubiquitous coverage~\cite{7827017}. In addition, it also enables a virtual FD mode even relying on half-duplex (HD) MNs. Using HD MNs instead of FD ones makes the monitoring system more cost-effective and less prone to self-interference. 

Nonetheless, a popular assumption in the current literature on wireless surveillance systems is that there is global CSI knowledge of all the untrusted links at the MNs and/or central processing unit (CPU). However, in practice, the MNs/CPUs typically have access only to imperfect CSI. Assuming perfect CSI for proactive monitoring is unrealistic and greatly affects the probability of successful monitoring. 
In~\cite{9381240}, the impact of channel uncertainty on proactive surveillance was investigated. In this work,  the authors formulated an optimization problem to enhance the surveillance performance under a covert constraint and showed that the uncertainty of the links can highly impact surveillance performance. In~\cite{10032289},  the authors analyzed the performance of a multi-antenna proactive monitoring system under imperfect instantaneous CSI knowledge of the untrusted link. Nevertheless,  a significant drawback of these studies is the lack of understanding on how to acquire the CSI of the untrusted links. 

Motivated by all the above,  we consider a CF-mMIMO surveillance system that monitors a pair of multi-antenna untrusted users using multiple multi-antenna MNs and propose an efficient CSI acquisition scheme. In our considered system, the MNs are HD and operate in either observing or jamming mode. Specifically, a group of MNs overhears the untrusted messages from the UT, while the remaining MNs send jamming signals to disrupt the UR.

The main contributions of our paper are:
\begin{itemize}
    \item We propose a complete transmission protocol for CF-mMIMO surveillance systems designed to monitor a pair of multi-antenna untrusted users with a simple and effective CSI acquisition approach. In our CF-mMIMO surveillance system, by leveraging pilot signals during both the uplink and downlink phases of the untrusted link and MMSE estimation techniques, the MNs estimate the effective channels to both the UT and UR. These channel estimates empower the MNs to design suitable methods to enhance the monitoring performance. Note that previous works either considered perfect CSI (e.g, \cite{9673103}) or imperfect CSI without specific channel estimation schemes (e.g, \cite{paperzahra, 9381240, 10032289}).   
    \item We derive analytical expressions for the SEs of the untrusted links and at the monitoring system, taking into account the proposed CSI acquisition and successive MMSE cancellation schemes. The MSP  is then derived. Specifically, we  investigate the monitoring performance of the CF-mMIMO surveillance system under two CSI knowledge scenarios;  namely \emph{1)} imperfect CSI knowledge at both the MNs and CPU \emph{2)} imperfect CSI knowledge at the MNs and no CSI knowledge at the CPU. 
    \item Numerical results demonstrate that CF-mMIMO surveillance systems, using the proposed CSI acquisition scheme, can offer competitive monitoring performance compared with an ideal system with perfect CSI knowledge of the untrusted communication link at both the MNs and CPU, regardless of the precoding design chosen for the untrusted transmission link.
\end{itemize}

\textit{Notation. } Throughout this paper, bold upper-case letters denote matrices whereas bold lower-case letters denote vectors; $(\cdot)^T$ and $(\cdot)^H$ stand for the matrix transpose and Hermitian transpose, respectively; $\mathbf{I}_M$ is the identity matrix, with size $M$; $||\cdot||$ and $|\cdot|$ are the Euclidean-norm and the absolute value operator; $\esp\{\cdot\}$ is the expectation operator, and $\operatorname{Var}(a)\triangleq\esp\left\{|a - \esp\{a\}|^2\right\}$. Finally, a zero mean circular symmetric complex Gaussian vector $\mathbf{x}$ with covariance matrix $\mathbf{C}$ is denoted by $\mathbf{x}\sim\mathcal{CN}(\mathbf{0}, \mathbf{C})$.

\section{System Model}
\begin{figure}[t]
    \centering
    \includegraphics[scale=0.3]{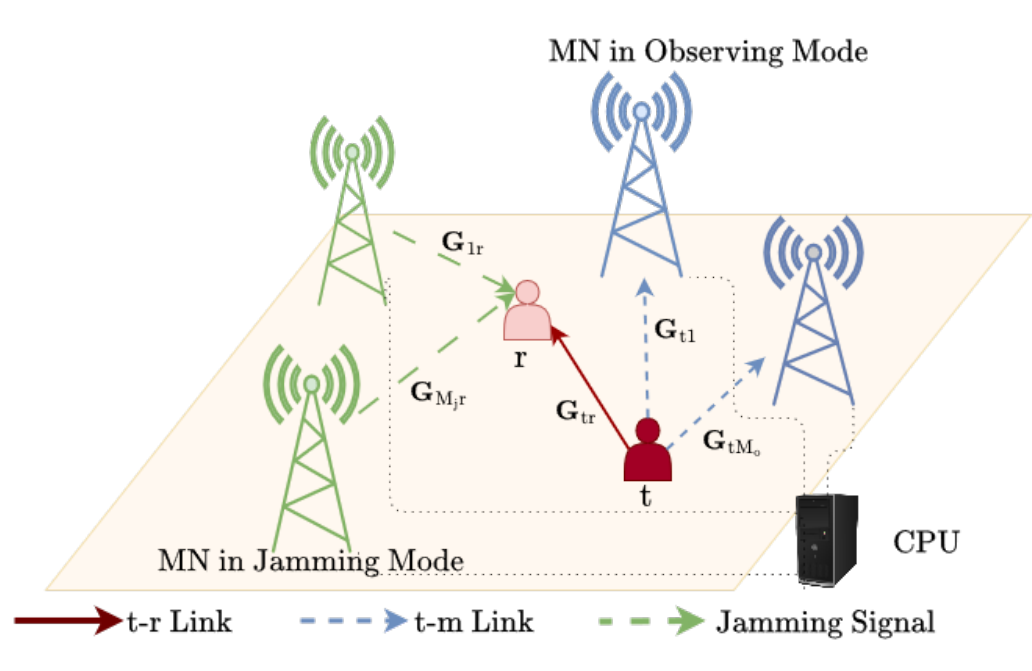}
    \caption{System model of a CF-mMIMO surveillance system with MNs operating in observing mode or jamming mode.}
    \label{fig:enter-label}
\end{figure}

As illustrated in Fig.~\ref{fig:enter-label}, the considered CF-mMIMO surveillance system comprises $M$ MNs and a communication untrusted pair. The MNs are equipped with $\NM$ antennas, while the  \tUT~and   \tUR \footnote{Throughout this work, we use the subscripts $\UT$, $\UR$, and  $\MN$ to refer to \tUT, \tUR, and  CPU, respectively.} are equipped with $\NT$ and $\NR$ antennas, respectively. We further assume reciprocity-based time division duplex (TDD) operation so that the uplink and downlink transmissions occur at different times but use the same frequency. The processing of the MNs is done in a centralized manner by a CPU, with the MNs being completely synchronized. Moreover, all nodes operate in HD mode. To monitor the untrusted pair, the MNs can switch between the observing mode, where they receive untrusted messages from \tUT, and jamming mode, where they send jamming signals to disrupt \tUR. We define $\alpha_m$ as a binary variable to indicate the operation assignment of each MN $m$, such that  MN $m$ operates in the jamming mode when $\alpha_m=0$  or it operates in the observing mode when $\alpha_m=1$. 

For the untrusted link, the \tUR~typically first transmits uplink pilots to the \tUT, enabling \tUT~to estimate the channel to \tUR. These channel estimates at the \tUT~are then used to define the precoder for transmitting the information signal to \tUR. On the other hand, to accurately detect the signals transmitted from the \tUT, the \tUR~requires knowledge of the channel gains, which can be obtained through beamforming downlink schemes~\cite{6736537}. Accordingly, by leveraging pilot signals transmitted during both the uplink training and downlink beamforming training phases of the untrusted link, the MNs can estimate the channels to both the \tUR~and \tUT. These channel estimates empower the MNs to design suitable methods to enhance their monitoring performance. The details of the uplink training, beamforming training, and downlink data transmission phases are provided in the next subsections.
\subsection{Uplink Training}\label{sec:upt}
During the uplink training, pilot sequences of length $\tau_\UR$ are sent from the \tUR~to \tUT, where the pilot sequence transmitted from the $\nr$th antenna of the \tUR~is denoted by $\phiu\in\mathbb{C}^{\tau_{\UR}\times1}$, with $||\phiu||^2=1$, $\forall\nr=1,\dots,\NR$. All pilot sequences are considered pairwise orthogonal, i.e., $\phiuh\bm{\varphi}_{\nrl}=0$ if $\nr\neq\nrl$. Thus, it is required that $\tau_{\UR}\geq\NR$. Simultaneously, the MNs also receive the pilot signals transmitted from the \tUR. Thus, the received $\NT\times\tau_{\UR}$ pilot matrix at the \tUT, and the $\NM\times\tau_{\UR}$ pilot matrix  at the $m$th MN are given by~\cite{9392380}
\vspace{-0.5em}
\begin{align}
   \mathbf{Y}_{\UTUR}&=\sqrt{\tau_{\UR} \rho_\UR} \sum_{\nr=1}^{\NR} \mathbf{g}_{\UTUR, \nr} \phiuh+\bm{\Omega}_{\UT},\\
    \mathbf{Y}_{m\UR}&=\sqrt{\tau_{\UR} \rho_\UR} \sum_{\nr=1}^{\NR} \mathbf{g}_{m\UR, \nr} \phiuh+\bm{\Omega}_{m},
\end{align}
respectively, where $\rho_\UR$ is the normalized power of each pilot symbol
transmitted by \tUR, while $\bm{\Omega}_{\UT}\in\mathbb{C}^{\NT\times\tau_{\UR}}$ and $\bm{\Omega}_{m}\in\mathbb{C}^{\NM\times\tau_{\UR}}$ are Gaussian matrices with independent and identically distributed (i.i.d.) $\mathcal{CN}(0,1)$ entries. Also, $\mathbf{g}_{\UTUR, \nr}$ and $\mathbf{g}_{m\UR, \nr}$ are the $\nr$th column of the channel matrix from the \tUR~to \tUT,   denoted by $\mathbf{G}_{\UTUR}$, and the $\nr$th column of the channel matrix from the \tUR~to the $m$th MN  denoted by $\mathbf{G}_{m\UR}$, respectively. 
Also, $\mathbf{G}_{\UTUR}\in\mathbb{C}^{\NT\times \NR}$ and $\mathbf{G}_{m\UR}\in\mathbb{C}^{\NM\times \NR}$ are modeled as
\begin{align}\label{eq:channeldef}
\mathbf{G}_{\UTUR}&=\sqrt{\beta_{\UTUR}}\mathbf{H}_{\UTUR},\\
\mathbf{G}_{m\UR}&=\sqrt{\beta_{m\UR}}\mathbf{H}_{m\UR},
\end{align}
respectively, where $\beta_{\UTUR}$ and $\beta_{m\UR}$ are large-scale fading coefficients, and $\mathbf{H}_{\UTUR}$ and $\mathbf{H}_{m\UR}$ are the associated small-scale fading matrices between the \tUR~and \tUT~and between the \tUR~and the $m$th MN, respectively. The entries of the small-scale fading matrices are assumed i.i.d. $\mathcal{CN}(0,1)$. The pilot sequences received at the \tUT~and at the $m$th MN are projected onto $\mathbf{Y}_{\UTUR}$ and $\mathbf{Y}_{m\UR}$, allowing the channel vector for each antenna of \tUR~to be estimated by \tUT~and MN $m$, as
\begin{align}
\Tilde{\mathbf{y}}_{\UTUR,\nr} & \triangleq\mathbf{Y}_{\UTUR} \phiu
=\sqrt{\tau_{\UR} \rho_{\UR}} \mathbf{g}_{\UTUR, \nr}+\bm{\Omega}_{\UT}\phiu, \label{eq:yRT}\\
\Tilde{\mathbf{y}}_{m\UR,\nr} & \triangleq\mathbf{Y}_{m\UR} \phiu
=\sqrt{\tau_{\UR} \rho_{\UR}} \mathbf{g}_{m\UR, \nr}+\bm{\Omega}_{m}\phiu. \label{eq:yRm}
\end{align}
Based on \eqref{eq:yRT} and \eqref{eq:yRm}, we can attain the estimates of the channel responses $\mathbf{g}_{\UTUR, \nr}$ and $\mathbf{g}_{m\UR, \nr}$ via a linear MMSE approach as
\begin{align}\label{eq:hatg}
    \hat{\mathbf{g}}_{\UTUR, \nr}\!&=\!\esp\left\{\mathbf{g}_{\UTUR, \nr}\Tilde{\mathbf{y}}_{\UTUR,\nr}^H\right\}\!\left(\esp\left\{\Tilde{\mathbf{y}}_{\UTUR,\nr}\Tilde{\mathbf{y}}_{\UTUR,\nr}^H\right\}\right)^{-1}\Tilde{\mathbf{y}}_{\UTUR,\nr}\nonumber\\
    &=\gamma_{\UTUR}\Tilde{\mathbf{y}}_{\UTUR,\nr},\\
    \hat{\mathbf{g}}_{m\UR, \nr}\!&=\!\esp\left\{\mathbf{g}_{m\UR, \nr}\Tilde{\mathbf{y}}_{m\UR,\nr}^H\right\}\!\left(\esp\left\{\Tilde{\mathbf{y}}_{m\UR,\nr}\Tilde{\mathbf{y}}_{m\UR,\nr}^H\right\}\right)^{-1}\Tilde{\mathbf{y}}_{m\UR,\nr}\nonumber\\
    &=\gamma_{m\UR}\Tilde{\mathbf{y}}_{m\UR,\nr},
\end{align}
where $\gamma_{\UTUR}\triangleq\frac{\sqrt{\tau_{\UR}\rho_{\UR}}\beta_{\UTUR}}{\tau_{\UR}\rho_{\UR}\beta_{\UTUR}+1}$ and $\gamma_{m\UR}\triangleq\frac{\sqrt{\tau_{\UR}\rho_{\UR}}\beta_{m\UR}}{\tau_{\UR}\rho_{\UR}\beta_{m\UR}+1}$. 
\subsection{Beamforming Training}
In this phase, the \tUT~beamforms the pilots using a precoding matrix derived from the channel estimate of \tUR~obtained during the uplink training phase,  $\mathbf{W}=\Big[\frac{\mathbf{w}_{1}}{||\mathbf{w}_1||},\dots,\frac{\mathbf{w}_{\NR}}{||\mathbf{w}_\NR||}\Big]$, $\mathbf{W} \in\mathbb{C}^{\NT\times\NR}$, with $\frac{\mathbf{w}_{\nr}}{||\mathbf{w}_\nr||}$, $\nr=1,\dots,\NR$, being the intended normalized $N_{\UT}\times1$ precoder vector for each antenna of the \tUR. 
 Let $\sqrt{\tau_{\UT}\rho_{\UT}}\phid\in\mathbb{C}^{\NR\times \tau_{\UT}}$ be the pilot sequence from the \tUT~to \tUR, with $\tau_{\UT}$ being the duration (in symbols) of the beamforming training, and $\rho_{\UT}$ being the maximum normalized transmit power at \tUT.  We assume that $\tau_{\UT}\geq\NR$, while the rows of $\phid$ are pairwise orthogonal, i.e., $\phid\phidh=\mathbf{I}_\NR$. Hence, the received pilot matrix at the \tUR~$\in \mathbb{C}^{\NR \times \tau_\UT}$, and at the $m$th MN $\in \mathbb{C}^{\NM \times \tau_\UT}$ are given by 
\begin{align}
    \mathbf{Y}_{\URUT}^T &= \sqrt{\tau_{\UT}\rho_{\UT}}\mathbf{G}_{\UTUR}^H\mathbf{W}\phid + \bm{\Omega}_{\UR}^T,\\
    \mathbf{Y}_{m\UT}^T &= \sqrt{\tau_{\UT}\rho_{\UT}}\mathbf{G}_{\UT m}^H\mathbf{W}\phid + \bm{\Omega}_{m}^T,
\end{align}
respectively, where $\mathbf{G}_{\UT m}\in\mathbb{C}^{\NT\times \NM}$ is the channel response between MN $m$ and \tUT, described as in \eqref{eq:channeldef}. As discussed in~\cite{6736537}, we can project $\phid$ onto $\mathbf{Y}_{\URUT}$ and $\mathbf{Y}_{m\UT}$, and use it to estimate the effective channels. Accordingly, 
\begin{align}
    \Tilde{\mathbf{Y}}_{\URUT}^T&\triangleq\mathbf{Y}_{\URUT}^T\phidh=\sqrt{\tau_{\UT}\rho_{\UT}}\mathbf{G}_{\UTUR}^H\mathbf{W}+\Tilde{\bm{\Omega}}_{\UR}^T,\label{eq:yrtp}\\
    \Tilde{\mathbf{Y}}_{m\UT}^T&\triangleq\mathbf{Y}_{m\UT}^T\phidh=\sqrt{\tau_{\UT}\rho_{\UT}}\mathbf{G}_{\UT m}^H\mathbf{W}+\Tilde{\bm{\Omega}}_{m}^T,\label{eq:ymtp}
\end{align}
where $\Tilde{\bm{\Omega}}_{\UR}^T\triangleq\bm{\Omega}_{\UR}^T\phidh$ and $\Tilde{\bm{\Omega}}_{m}^T\triangleq\bm{\Omega}_{m}^T\phidh$. Let us define $\mathbf{A}_\UR\triangleq\mathbf{G}_{\UTUR}^H\mathbf{W}$, with entries given by $a_{\nr,\nrl}\triangleq\mathbf{g}_{\UTUR,\nr}^H \mathbf{w}_{\nrl}$, and $\mathbf{B}_m=[\mathbf{b}_{1},\dots,\mathbf{b}_{\NM}]^H$ with $\bmest^H\triangleq\mathbf{g}_{\UT m,\nm}^H\mathbf{W}$, $\forall~p=1,\dots,\NM$. Then, the MMSE channel estimate of $\bmest$ is given according to the following proposition. 

\begin{proposition}
Let $\ymtest$ denote the $p$th column of $ \Tilde{\mathbf{Y}}_{m\UT}$, the MMSE channel estimate of $\bmest$ is written as
\begin{align}
    \hat{\mathbf{b}}_{\nm}&=\esp\left\{\bmest\right\}+\sqrt{\tau_\UT\rho_\UT}\mathbf{C}_{\bmest,\bmest}\left(\tau_\UT\rho_\UT\mathbf{C}_{\bmest,\bmest}+\mathbf{I}_{\NR}\right)^{-1}\nonumber\\
    &\times\left(\ymtest-\sqrt{\tau_{\UT} \rho_{\UT}} \esp\left\{\bmest\right\}\right),\label{eq:hatbm}
\end{align}
\end{proposition} 
\begin{proof}\label{proof1}
 The proof is provided in 
 Appendix A.
\end{proof}
\subsection{Untrusted Data Transmission}
 Let $\mathbf{x}_{\URUT}\in\mathbb{C}^{\NR\times1}$, with $\esp\left\{\mathbf{x}_{\URUT}\mathbf{x}_{\URUT}^H\right\}=\mathbf{I}_\NR$, be the symbol vector intended to \tUR. \tUT~uses the channel estimate obtained in the uplink training phase to 
precode the symbols, and then it transmits the precoded signal
vector to the \tUR. The transmitted signal from the \tUT~is written as
\begin{align}\label{eq:st}
    \mathbf{s}_{\UT} = \sqrt{\rho_{\UT}}\mathbf{W}\bm{\Lambda}_\UR^{1/2}\mathbf{x}_{\URUT},
\end{align}
where $\bm{\Lambda}_\UR$ is a diagonal matrix whose diagonal elements are $\lambda_1,\dots,\lambda_\NR$, set to satisfy $\esp\left\{||\mathbf{s}_{\UT}||^2\right\}=\rho_{\UT}$.
At the same time, the MNs operating in the jamming mode send jamming signals to interfere with the untrusted communication link. Let $\mathbf{x}^\mathrm{J}\in\mathbb{C}^{\NR\times1}$, with $\esp\left\{(\mathbf{x}^\mathrm{J})(\mathbf{x}^{\mathrm{J}})^H\right\}=\mathbf{I}_\NR$, denote the jamming symbol intended to the \tUR. We assume that the MNs employ a maximum-ratio (MR) precoding technique for jamming transmission to maximize the jamming power received at \tUR. Thus,  the signal vector transmitted by the $m$th MN in the jamming mode, $\mathbf{s}_{m}^\mathrm{J} \in\mathbb{C}^{\NM\times1}$, can be written as
\begin{align}\label{eq:sigj}
    \mathbf{s}_{m}^\mathrm{J} = (1-\alpha_m)\sqrt{\rho_\mathrm{J}}\hat{\mathbf{G}}_{m\UR}\bm{\Pi}_{m\UR}^{1/2}\mathbf{x}^\mathrm{J},
\end{align}
where $\rho_\mathrm{J}$ is the maximum normalized transmit power at the MNs in jamming mode, $\bm{\Pi}_{m\UR}$ is a diagonal matrix with diagonal elements given by $\pi_{m,1},\dots,\pi_{m,\NR}$, chosen to satisfy $\esp\{||\mathbf{s}_{m}^\mathrm{J}||^2\}\leq\rho_\mathrm{J}$  for each MN in jamming mode. Thus, given the transmitted signal $\mathbf{s}_{\UT}$ in \eqref{eq:st} and $\mathbf{s}_{m}^\mathrm{J}$ in \eqref{eq:sigj}, the received signal at the \tUR, and at the $m$th MN in observing mode, are written, respectively, as
\vspace{-1em}
\begin{align}
    \mathbf{y}_{\UR} &= \alpha_m\mathbf{G}_{\UTUR}^H\mathbf{s}_{\UT}+\sum_{m=1}^M\mathbf{G}_{m\UR}^H \mathbf{s}_{m}^\mathrm{J}+\bm{\omega}_{\UR}\nonumber\\
    &=\alpha_m\sqrt{\rho_{\UT}}\mathbf{G}_{\UTUR}^H\mathbf{W}\bm{\Lambda}_\UR^{1/2}\mathbf{x}_{\URUT}+\sum_{m=1}^M\!(1-\alpha_m)\sqrt{\rho_\mathrm{J}}\nonumber\\
    &\times\mathbf{G}_{m\UR}^H\hat{\mathbf{G}}_{m\UR}\bm{\Pi}_{m\UR}^{1/2}\mathbf{x}^\mathrm{J}+\bm{\omega}_{\UR},\label{eq:sigUR}\\
    \mathbf{y}_{m} &= \alpha_m\mathbf{G}_{\UT m}^H\mathbf{s}_{\UT}+\sum_{m'=1}^M\mathbf{G}_{mm'}^H \mathbf{s}_{m'}^\mathrm{J}+\bm{\omega}_{m}\nonumber\\
    &= \alpha_m\sqrt{\rho_{\UT}}\mathbf{G}_{\UT m}^H\mathbf{W}\bm{\Lambda}_\UR^{1/2}\mathbf{x}_{\URUT}+\sum_{m'=1}^M\!(1-\alpha_{m'})\sqrt{\rho_\mathrm{J}}\nonumber\\
    &\times\mathbf{G}_{mm'}^H\hat{\mathbf{G}}_{m'\UR}\bm{\Pi}_{m'\UR}^{1/2}\mathbf{x}^\mathrm{J}+\bm{\omega}_{m},\label{eq:sigmt}
\end{align}
where $\bm{\omega}_{\UR}$  and $\bm{\omega}_{m}$ are the $\NR\times1$  and $\NM\times1$ noise vectors at the \tUR~and at $m$th MN, respectively. Also, $\mathbf{G}_{mm'}$ denotes the channel matrix between MN $m$ and MN $m'$. The elements of $\mathbf{G}_{mm'}$ are assumed i.i.d. $\mathcal{CN}(0, \beta_{mm'})$ for $m'\neq m$, whereas for $m'=m$, $\mathbf{G}_{mm'}=\mathbf{0}$, $\forall m$. 
The effective channel estimate  ${\mathbf{B}}_m$ is combined by the $m$th MN  to detect $\mathbf{x}_{\URUT}$. Specifically,   an MMSE combining matrix is designed as
\begin{align}\label{eq:vmtil}
    {\mathbf{V}}_m=\hat{\mathbf{B}}_m\big(\hat{\mathbf{B}}_m^H\hat{\mathbf{B}}_m+\varrho\mathbf{I}_{\NR}\big)^{-1},
\end{align}
where $\varrho$ is the per stream signal-to-noise ratio (SNR). Finally, the aggregated received signal   for observing the untrusted link at the CPU can be obtained as
\vspace{0.5cm}
\begin{align}\label{eq:tildexrt}
    {\mathbf{z}}_{\MN}&=\sum_{m=1}^{M}\alpha_m{\mathbf{V}}_m^H\mathbf{y}_{m}
    \nonumber\\
    &=\sum_{m=1}^{M}\!\alpha_m\Big(\!\sqrt{\rho_{\UT}}{\mathbf{V}}_m^H\mathbf{G}_{\UT m}^H\mathbf{W}\bm{\Lambda}_\UR^{1/2}\mathbf{x}_{\URUT}+{\mathbf{V}}_m^H\bm{\omega}_{m}\nonumber\\
    &+{\mathbf{V}}_m^H\!\sum_{m'=1}^M\!(1-\alpha_{m'})\sqrt{\rho_\mathrm{J}}
    \mathbf{G}_{mm'}^H\hat{\mathbf{G}}_{m'\UR}\bm{\Pi}_{m'\UR}^{1/2}\mathbf{x}^\mathrm{J}\Big).
\end{align}
\section{Spectral Efficiency}\label{sec:se}
In this section, we derive the SE expressions for the untrusted communication link and the  surveillance system. First, a general formula for the SE with MMSE detection given arbitrary side information (which is  independent of the transmit signals)  is provided.  Subsequently, a closed-form SE expression at the \tUR~is derived, assuming perfect knowledge of the CSI of \tUT. For the surveillance system, two scenarios are investigated: \emph{1)} imperfect CSI knowledge at the MNs and no CSI
knowledge at the CPU; \emph{2)} imperfect CSI knowledge at both the MNs
and CPU.
\begin{proposition}
Given the received signal at the \tUR~and at the $m$th MN, in \eqref{eq:sigUR} and \eqref{eq:sigmt}, respectively, the achievable SE at R and at the CPU assuming MMSE detection can be written as
\begin{align}
    \mathrm{SE}_\UR&=\Big(\!1\!-\!\frac{\tau_\UT+\tau_\UR}{\tau}\!\Big)\esp\Bigl\{\log_2\left(\det\left(\mathbf{I}_{N_\UR}+\bm{\Upsilon}_\UR\right)\right)\Bigr\},\label{eq:SEr}\\
    \mathrm{SE}_\MN&=\Big(\!1\!-\!\frac{\tau_\UT+\tau_\UR}{\tau}\!\Big)\esp\Bigl\{\log_2\left(\det\left(\mathbf{I}_{N_\UR}+\bm{\Upsilon}_\MN\right)\right)\Bigr\},\label{eq:SEm}
\end{align}
where $\tau$ is the coherence interval, $\bm{\Upsilon}_\UR$ and $\bm{\Upsilon}_{\MN}$ are given by
\begin{align}
    \bm{\Upsilon}_\UR&=  \rho_\UT\esp\big\{\bm{\Lambda}_\UR^{1/2}\mathbf{A}_\UR^H|\bm{\Theta}_\UR\big\}(\bm{\Psi}_{\UR})^{-1}\esp\big\{\mathbf{A}_\UR\bm{\Lambda}_\UR^{1/2}|\bm{\Theta}_\UR\big\},\label{eq:upsilonr}\\
    \bm{\Upsilon}_{\MN}&=  \rho_\UT\esp\big\{\mathbf{D}_m^H|\bm{\Theta}_{\MN}\big\}(\bm{\Psi}_{\MN})^{-1}\esp\big\{\mathbf{D}_m|\bm{\Theta}_{\MN}\big\},
\end{align}
where $\bm{\Theta}_{\UR}$ and $\bm{\Theta}_{\MN}$ represent the side information, independent of $\mathbf{x}_{\URUT}$, and $\mathbf{D}_m$, $\bm{\Psi}_{\UR}$ and $\bm{\Psi}_{\MN}$ are given by  
\begin{align}
    \mathbf{D}_m&\triangleq\sum_{m=1}^M\alpha_m\mathbf{V}_m^H\mathbf{B}_m\bm{\Lambda}_\UR^{1/2},\\
    \bm{\Psi}_{\UR}&=\mathbf{I}_{N_\UR}+\rho_{\mathrm{J}}\esp\big\{\Frj(\Frj)^H\big\}+\rho_\UT\esp\big\{\mathbf{A}_\UR\bm{\Lambda}_\UR\mathbf{A}_\UR^H|\bm{\Theta}_\UR\big\}\nonumber\\
    &-\rho_\UT\esp\big\{\mathbf{A}_\UR\bm{\Lambda}_\UR^{1/2}|\bm{\Theta}_\UR\big\}\esp\big\{\mathbf{A}_\UR\bm{\Lambda}_\UR^{1/2}|\bm{\Theta}_\UR\big\}, \\
    \bm{\Psi}_{\MN}&=\rho_{\mathrm{J}}\esp\left\{\sum_{m=1}^{M}\sum_{l=1}^{M}\alpha_m\alpha_{l}\mathbf{V}_m^H\Fmj_m(\Fmj_l)^H\mathbf{V}_{l}|\bm{\Theta}_{\MN}\right\}\nonumber\\
    &+\rho_\UT\esp\Bigl\{\mathbf{D}_m\mathbf{D}_m^H|\bm{\Theta}_{\MN}\!\Bigr\}\!+\esp\Bigl\{\sum_{m=1}^{M}\sum_{l=1}^{M}\!\!\alpha_m\alpha_l{\mathbf{V}}_m^H{\mathbf{V}}_l|\bm{\Theta}_{\MN}\!\Bigr\}\nonumber\\
    &-\rho_\UT\esp\big\{\mathbf{D}_m|\bm{\Theta}_{\MN}\big\}\esp\big\{\mathbf{D}_m^H|\bm{\Theta}_{\MN}\big\},
\end{align}
respectively, where 
\begin{align}
    \Frj&\triangleq\sum_{m=1}^M\!(1-\alpha_m)\sqrt{\rho_\mathrm{J}}\mathbf{G}_{m\UR}^H\hat{\mathbf{G}}_{m\UR}\bm{\Pi}_{m\UR}^{1/2},\\
    \Fmj_m&\triangleq\!\sum_{m'=1}^M(1-\alpha_{m'})\mathbf{G}_{mm'}^H\hat{\mathbf{G}}_{m'\UR}\bm{\Pi}_{m'\UR}^{1/2}.
\end{align}
\end{proposition}
\begin{proof}
   The proof is provided in Appendix B. 
\end{proof}
Now, for the untrusted link, we examine the SE performance when the \tUR~has perfect knowledge of the effective untrusted channel, which represents the worst-case scenario from a monitoring performance perspective. Consequently, we can obtain the following closed-form expression for the SE.
\begin{proposition}
The SE at the \tUR, assuming perfect knowledge of the effective untrusted channel, is given by
\vspace{-0.5em}
\begin{align}\label{eq:ser}
    \mathrm{SE}_\UR=\Big(1-\frac{\tau_\UT+\tau_\UR}{\tau}\Big)\esp\Big\{\log_2\Big(1+\sum_{\nr=1}^{\NR}\Gamma_\nr\Big)\Big\},
\end{align}
where $\Gamma_\nr$ is the signal-to-interference-plus-noise ratio (SINR) at the $\nr$th receive antenna of \tUR, given by
    \begin{align}\label{eq:sinrnr}
         \Gamma_\nr = \frac{\rho_\UT\lambda_{\nr}|a_{\nr,\nr}|^2}{1+\rho_\UT\lambda_{\nrl}\sum_{\nrl=1 \atop \nrl\neq \nr}^{\NR}|a_{\nr,\nrl}|^2+\rho_{\mathrm{J}}\mathcal{I}},
    \end{align}
    \vspace{-1em}
    with 
    \vspace{-1em}
    \begin{align}
        \mathcal{I}&=\NM\sum_{\nrl=1}^\NR\sum_{m=1}^{M}(1-\alpha_m){\pi}_{m,\nrl}\gamma_{m\UR}\nonumber\\
        &+\NM^2\left(\sum_{m=1}^{M}(1-\alpha_m)\sqrt{\pi_{m,\nr}}\gamma_{m\UR}\right)^2.
    \end{align}
\end{proposition}
\begin{proof}
   The proof follows similar steps as \cite[Appendix A]{paperzahra}.
\end{proof}
\vspace{-0.6em}
In addition, for the surveillance system, we consider two scenarios based on the availability of side information, as outlined in the next subsections.\footnote{In Section \ref{sec:case1}, the superscript (1) stands for the scenario with imperfect CSI at the MNs and no CSI knowledge at the CPU, while in Section~\ref{sec:case2}, the superscript (2) stands for the scenario with imperfect CSI knowledge at the MNs and at the CPU.}
\vspace{-0.5em}
\subsection{Imperfect CSI Knowledge at the MNs and No CSI Knowledge at the CPU}~\label{sec:case1}
In this case, even though the estimates of the effective channel between the MNs and the untrusted nodes are computed at each MN, we assume that this information is not forwarded to the CPU, hence $\bm{\Theta}_{\MN}^{(1)}=\varnothing$. Thus, \eqref{eq:SEm} is rewritten as
\begin{align}\label{eq:se1}
     \mathrm{SE}_{\MN}^{(1)}=\Big(1-\frac{\tau_\UT+\tau_\UR}{\tau}\Big)\log_2\left(\det\left(\mathbf{I}_{N_\UR}+\bm{\Upsilon}_{\MN}^{(1)}\right)\right),
\end{align}
where $\bm{\Upsilon}_{\MN}^{(1)}=\rho_\UT\esp\big\{\mathbf{D}_m^H\big\}(\bm{\Psi}_{\MN}^{(1)})^{-1}\esp\big\{\mathbf{D}_m\big\}$
with  
\vspace{-0.5em}
\begin{align}
    \bm{\Psi}_{\MN}^{(1)}\!&=\rho_{\mathrm{J}}\esp\left\{\sum_{m=1}^{M}\sum_{l=1}^{M}\alpha_m\alpha_{l}\mathbf{V}_m^H\Fmj_m(\Fmj_l)^H\mathbf{V}_{l}\!\right\}\nonumber\\&+\!\rho_\UT\esp\Bigl\{\mathbf{D}_m\mathbf{D}_m^H\!\Bigr\}+\esp\left\{\sum_{m=1}^{M}\sum_{l=1}^{M}\alpha_m\alpha_l{\mathbf{V}}_m^H{\mathbf{V}}_l\right\}\nonumber\\
    &-\rho_\UT\esp\big\{\mathbf{D}_m\big\}\esp\big\{\mathbf{D}_m^H\big\}.
\end{align}
\subsection{Imperfect CSI Knowledge at the MNs and at the CPU}~\label{sec:case2}
In this scenario, we assume that the estimates of the CSI of the \tUT~computed by the MNs in observing mode in \eqref{eq:hatbm} are forwarded to the CPU. Therefore, $\bm{\Theta}_{\MN}^{(2)}=[\hat{\mathbf{B}}_1,\dots,\hat{\mathbf{B}}_{M}]$. Note that the elements of $\mathbf{B}_m$ are Gaussian distributed, and thus the MMSE estimates $\hat{\mathbf{B}}_m$ and corresponding estimation error $\Tilde{\mathbf{B}}_m\triangleq{\mathbf{B}}_m-\hat{\mathbf{B}}_m$ are uncorrelated and independent.  Let $\hat{\mathbf{D}}_m\triangleq\sum_{m=1}^M\alpha_m\mathbf{V}_m^H\hat{\mathbf{B}}_m\bm{\Lambda}_\UR^{1/2}$. Hence, using~\eqref{eq:SEm}, the achievable SE at 
the CPU  can be obtained as
\begin{align}\label{eq:se2}
     \mathrm{SE}_{\MN}^{(2)}\!=\!\Big(1-\frac{\tau_\UT+\tau_\UR}{\tau}\Big)\esp\big\{\log_2\big(\det\big(\mathbf{I}_{N_\UR}+\bm{\Upsilon}_{\MN}^{(2)}\big)\big)\big\},
\end{align}
where $\bm{\Upsilon}_{\MN}^{(2)}=\rho_\UT\hat{\mathbf{D}}_m^H(\bm{\Psi}_{\MN}^{(2)})^{-1}\hat{\mathbf{D}}_m$ 
with
\begin{align}
\bm{\Psi}_{\MN}^{(2)}&=\rho_{\mathrm{J}}\sum_{m=1}^{M}\sum_{l=1}^{M}\alpha_m\alpha_l{\mathbf{V}}_m^H\esp\big\{\Fmj_m(\Fmj_l)^H\big\}{\mathbf{V}}_l\nonumber\\
    &+\sum_{m=1}^{M}\sum_{l=1}^{M}\alpha_m\alpha_l{\mathbf{V}}_m^H{\mathbf{V}}_l+\rho_\UT\esp\big\{\Tilde{\mathbf{D}}_m\Tilde{\mathbf{D}}_m^H\big\},
\end{align}
and  $\Tilde{\mathbf{D}}_m\triangleq\mathbf{D}_m-\hat{\mathbf{D}}_m$.
\subsection{Monitoring Success Probability}
The CPU can reliably monitor the untrusted communication between the \tUT~and \tUR~if $\mathrm{SE}_{\MN}^{(\mathrm{f})}\geq\mathrm{SE}_\UR, \mathrm{f} \in \{1,2\}$, with $\mathrm{SE}_{\MN}^{(1)}$, $\mathrm{SE}_{\MN}^{(2)}$ and $\mathrm{SE}_\UR$ given by \eqref{eq:se1}, \eqref{eq:se2} and \eqref{eq:ser}, respectively. Thus, the MSP can be computed as~\cite{paperzahra} 
\begin{align}
    \mathrm{MSP}^{(\mathrm{f})}=\Pr(\mathrm{SE}_{\MN}^{(\mathrm{f})}\geq\mathrm{SE}_\UR). 
\end{align}
\section{Numerical Results}\label{sec:Results}
In this section, the performance of the considered CF-mMIMO surveillance system is evaluated in terms of the MSP. Zero-forcing (ZF) and  maximum-ratio transmission (MRT) precoding techniques are considered for the transmit precoding matrix $\mathbf{W}$ design at the \tUT. The MNs, \tUT, and \tUR~are uniformly distributed within a $D\times D$ km$^2$ area. The wrapped-around technique is used to avoid the boundary effects. The values for bandwidth, transmit pilot power, pilot length, large-scale fading, and noise are extracted from~\cite{7827017}. Moreover, $\NT=\NR=4$, the transmit power for each MN in jamming mode is set as $\rho_{\mathrm{J}}=1$ W, and the per stream SNR is $\varrho=1/\rho_\UT$.

\begin{figure}[t]
    \centering
    \includegraphics[scale=0.37]{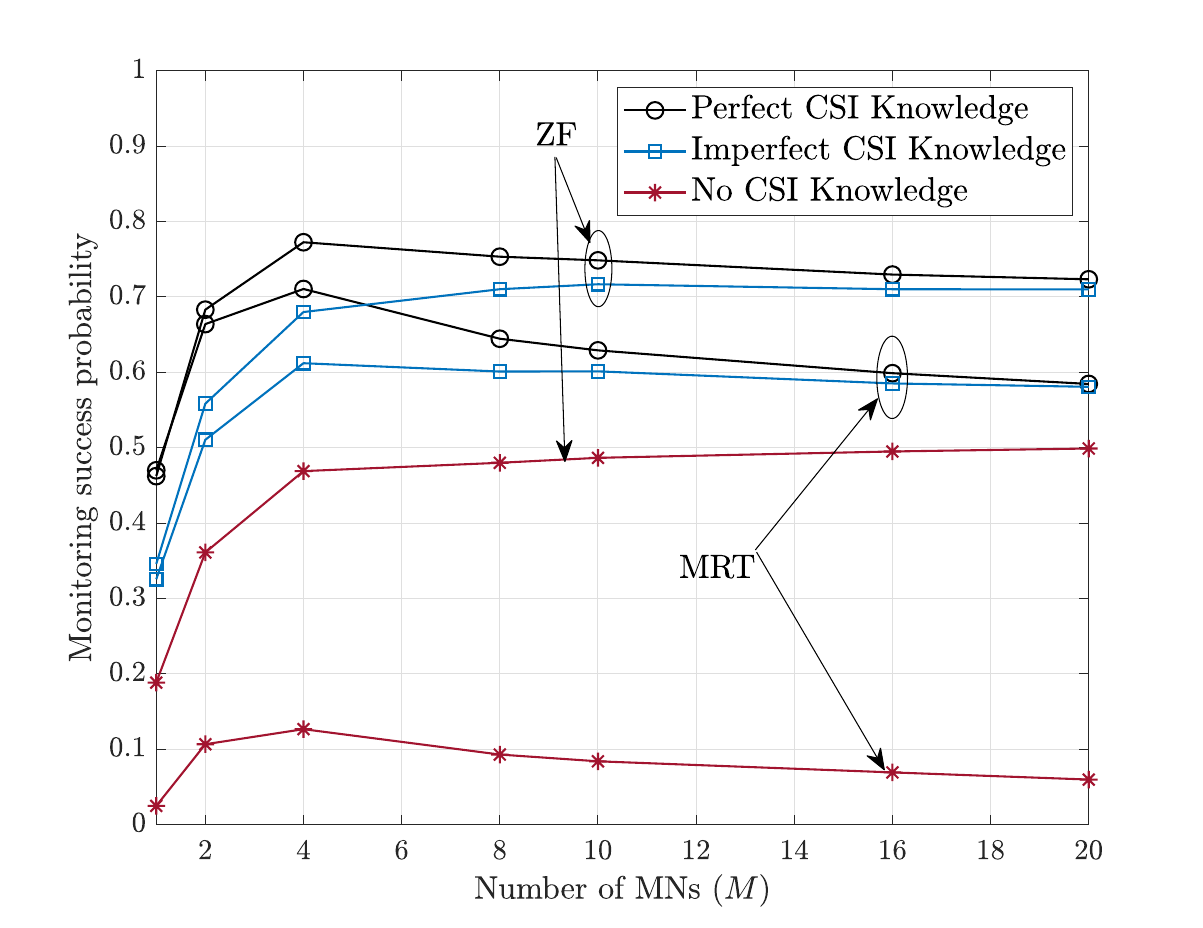}
    \caption{Monitoring success probability  versus the number of monitoring nodes, $M$, with $N_{\mathrm{MT}}=240$.}
    \label{fig:nmns}
\end{figure}
\balance
Figure~\ref{fig:nmns} illustrates the MSP versus the number of MNs for different scenarios of CSI availability, namely imperfect CSI knowledge at the MNs and no CSI knowledge at the CPU, and imperfect CSI knowledge at both the MNs and  CPU. In Fig.~\ref{fig:nmns}, the total number of antennas at all the MNs is fixed as $N_{\mathrm{MT}}=240$ and $D=1$ km. For performance comparison, we also illustrate the results for the ideal monitoring scenario, with perfect CSI knowledge of the untrusted communication link at both the MNs and   CPU. As expected,   perfect CSI knowledge presents an upper bound on the MSP performance. Nonetheless, the performance gap between this scenario and other scenarios relying on the proposed CSI acquisition approach is small, especially for a high number of MNs. This highlights the effectiveness of our proposed acquisition approaches.
Moreover, we observe that as the number of MNs in the surveillance system increases, the MSP remains relatively stable for all cases. This is due to the fact that increasing $M$, and, hence, decreasing $\NM$, has two effects on the MSP performance 1) decreases the diversity and array gains, and 2) increases the macro-diversity gain and decreases the path loss. 
Furthermore, it can be observed that the CF-mMIMO system demonstrates superior monitoring performance in scenarios where the \tUT~utilizes the ZF precoding design compared to those employing the MRT design.  In addition, the CF-mMIMO surveillance system with imperfect CSI knowledge at both the MNs and  CPU  
provides performance gains of up to $200\%$ and $450\%$ for ZF and MRT precoding design at the \tUT, respectively, in comparison to the system with imperfect CSI knowledge at MNs and no CSI knowledge at the CPU.


\begin{figure}[t]
    \centering
\includegraphics[scale=0.32]{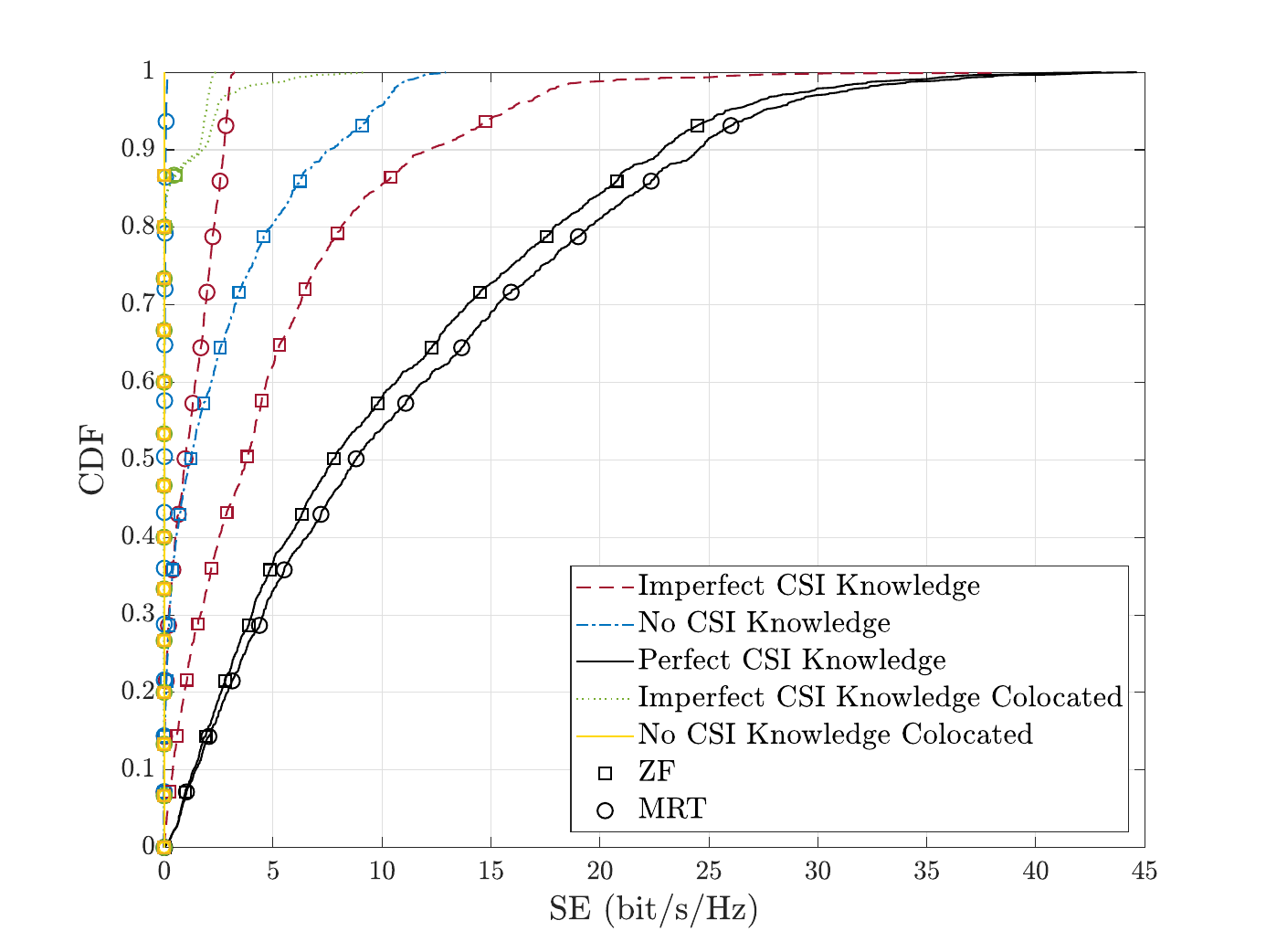}
    \vspace{-0.45em}
    \caption{CDF of the  SE at the CPU, with $M=4$ and $\NM=60$.}
    \label{fig:cdf}
      \vspace{-0.5cm}
\end{figure}

Figure~\ref{fig:cdf} illustrates the cumulative distribution function (CDF) of the average SE at the CPU. For performance comparison, the results of a CF-mMIMO surveillance system with perfect CSI knowledge at the CPU and a colocated mMIMO surveillance system are also presented. Here,  $D = 0.3$ km and  the residual self-interference for the colocated mMIMO surveillance system  is assumed as $30$ dB. Also, half of the antennas are employed for jamming purposes, and the rest are used for observing the untrusted communication link. Figure~\ref{fig:cdf} also shows that the monitoring performance attained with a CF-mMIMO surveillance system along with our proposed CSI acquisition scheme significantly outperforms the colocated mMIMO surveillance system. This is due to the fact that the CF-mMIMO system benefits from higher macro-diversity and lower path loss, while the colocated system suffers from self-interference between the transmit and receive antennas. 

\section{Conclusions}\label{sec:Conclusions}
We investigated the CSI acquisition of a multi-antenna untrusted link in a CF-mMIMO surveillance system relying on multiple multi-antenna MNs.
In our proposed scheme, by exploiting the pilot signals during both the uplink and downlink phases of the untrusted link and MMSE estimation, the MNs are capable of estimating the effective channels of both the  UT and UR. The monitoring performance of the CF-mMIMO surveillance system was evaluated for two scenarios: \emph{1}) imperfect CSI knowledge at both the MNs and CPU; \emph{2}) imperfect CSI knowledge at the MNs and no CSI knowledge at the CPU. Simulation results showed that the proposed CSI acquisition scheme is effective.  Moreover,  we confirmed that the  CF-mMIMO surveillance system along with our proposed CSI acquisition scheme significantly outperforms the colocated mMIMO surveillance system.  
\section*{Appendix A \\Proof of Proposition~1}
The MMSE estimate of $\bmest$ can be computed as
\begin{align}\label{eq:bma}
\hat{\mathbf{b}}_{\nm}&=\esp\bigl\{\bmest\bigr\}+\mathbf{C}_{\bmest,\ymtest}\mathbf{C}_{\ymtest\ymtest}^{-1}\left(\ymtest-\esp\bigl\{\ymtest\bigr\}\right),
\end{align}
where $\mathbf{C}_{\bmest,\ymtest}$ is given by
\begin{align}\label{eq:cbm1}
    \mathbf{C}_{\bmest,\ymtest}=\esp\left\{\left(\bmest-\esp\bigl\{\bmest\bigr\}\right)\left(\ymtest-\esp\bigl\{\ymtest\bigr\}\right)^T\right\}. 
\end{align}
Considering $\omegam$ as the $p$th column of $\Tilde{\bm{\Omega}}_{m}$, which is i.i.d. $\mathcal{CN}(0,1)$ and independent of $\bmest$, \eqref{eq:cbm1} can be rewritten as
\begin{align}\label{eq:cbm2}
   \mathbf{C}_{\bmest,\ymtest} &=\sqrt{\tau_\UT\rho_\UT}\esp\Bigl\{\bmest\bmest^T\!-\!\bmest\esp\bigl\{\bmest^T\bigr\}\!-\!\esp\{\bmest\}\bmest^T\nonumber\\&+\!\esp\bigl\{\bmest\bigr\}\esp\bigl\{\bmest^T\bigr\}\Bigr\}=\sqrt{\tau_\UT\rho_\UT}\mathbf{C}_{\bmest,\bmest}.
\end{align}
Analogously to \eqref{eq:cbm1}, $\mathbf{C}_{\ymtest,\ymtest}$ can be rewritten as
\begin{align}\label{eq:cym}
     \mathbf{C}_{\ymtest,\ymtest}&=\esp\left\{\left(\ymtest-\esp\bigl\{\ymtest\bigr\}\right)\left(\ymtest-\esp\bigl\{\ymtest\bigr\}\right)^T\right\} \nonumber\\
    &=\tau_\UT\rho_\UT\esp\Bigl\{\bmest\bmest^T\!-\!\bmest\esp\bigl\{\bmest^T\bigr\}\!-\!\esp\bigl\{\bmest\bigr\}\bmest^T\nonumber\\
    &+\!\esp\bigl\{\bmest\bigr\}\esp\bigl\{\bmest^T\bigr\}\Bigr\}\!+\!\esp\Bigl\{\omegam\omegam^T\!\!\!+\!\esp\bigl\{\omegam\bigr\}\esp\bigl\{\omegam^T\bigr\}\Bigr\}\nonumber\\
    &=\tau_\UT\rho_\UT\mathbf{C}_{\bmest,\bmest}\!+\!\mathbf{I}_{\NR}.
\end{align}
By replacing \eqref{eq:cbm2} and \eqref{eq:cym} into \eqref{eq:bma} we obtain \eqref{eq:hatbm}. 
\balance
\section*{Appendix B \\Proof of Proposition~2}
By denoting the differential entropy as $h(\cdot)$, the mutual information between $\mathbf{x}_\URUT$ and $\mathbf{y}_\UR$ is defined as
\begin{align}\label{eq:ixyr}
    I(\mathbf{x}_\URUT;\mathbf{y}_\UR,\bm{\Theta}_\UR)=h(\mathbf{x}_\URUT|\bm{\Theta}_\UR)-h(\mathbf{x}_\URUT|\mathbf{y}_\UR,\bm{\Theta}_\UR).
\end{align}
Following~\cite[Appendix C]{9079911}, under Gaussian signaling, we obtain $h(\mathbf{x}_\URUT|\bm{\Theta}_\UR)=\log_2(\det(\pi e \mathbf{I}_\NR))$. Next, following~\cite[Appendix I]{1624653}, $h(\mathbf{x}_\URUT|\mathbf{y}_\UR,\bm{\Theta}_\UR)$ is upper bounded by
\begin{align}
    h(\mathbf{x}_\URUT|\mathbf{y}_\UR,\bm{\Theta}_\UR)&\leq\esp\Bigl\{\log_2\left(\det\left(\pi e \esp\bigl\{\epsilon_\URUT \epsilon_\URUT^H|\bm{\Theta}_\UR\bigr\}\right)\right)\Bigr\},\label{eq:hxy}
\end{align}
where $\epsilon_\URUT$ is the MMSE estimation error of $\mathbf{x}_\URUT$ given $\mathbf{y}_\UR$ and $\bm{\Theta}_\UR$. Accordingly, $\esp\{\epsilon_\URUT \epsilon_\URUT^H|\bm{\Theta}_\UR\}$ is computed as 
\begin{align}~\label{eq:xuruthat}
    \esp\{\epsilon_\URUT \epsilon_\URUT^H|\bm{\Theta}_\UR\}&=\mathbf{C}_{\mathbf{x}_\URUT\mathbf{x}_\URUT|\bm{\Theta}_\UR}-\mathbf{C}_{\mathbf{x}_\URUT\mathbf{y}_\UR|\bm{\Theta}_\UR}\mathbf{C}_{\mathbf{y}_\UR\mathbf{y}_\UR|\bm{\Theta}_\UR}^{-1}\mathbf{C}_{\mathbf{y}_\UR\mathbf{x}_\URUT|\bm{\Theta}_\UR}.
\end{align}

The covariance matrices in~\eqref{eq:xuruthat} are calculated as
\begin{align}
\mathbf{C}_{\mathbf{x}_\URUT\mathbf{x}_\URUT|\bm{\Theta}_\UR}&=\esp\Bigl\{\mathbf{x}_\URUT\mathbf{x}_\URUT^H|\bm{\Theta}_\UR\Bigr\}\!=\!\mathbf{I}_\NR,\label{eq:cxurxur}\\
    \mathbf{C}_{\mathbf{x}_\URUT\mathbf{y}_\UR|\bm{\Theta}_\UR}&=\esp\Bigl\{\mathbf{x}_\URUT\mathbf{y}_\UR^H|\bm{\Theta}_\UR\Bigr\}\!=\!\sqrt{\rho_\UT}\esp\Bigl\{\bm{\Lambda}_\UR^{1/2}\!\mathbf{A}_\UR^H|\bm{\Theta}_\UR\Bigr\},\label{eq:cxuryr}\\
    \mathbf{C}_{\mathbf{y}_\UR\mathbf{y}_\UR|\bm{\Theta}_\UR}&=\esp\Bigl\{\mathbf{y}_\UR\mathbf{y}_\UR^H|\bm{\Theta}_\UR\Bigr\}\!=\!\mathbf{I}_{N_\UR}\!\!+\!\rho_{\mathrm{J}}\esp\Bigl\{\Frj(\Frj)^H\!\Bigr\}\nonumber\\\!&+\!\rho_\UT\esp\Bigl\{\mathbf{A}_\UR\bm{\Lambda}_\UR\mathbf{A}_\UR^H|\bm{\Theta}_\UR\Bigr\},\label{eq:cyryr}\\
   \mathbf{C}_{\mathbf{y}_\UR\mathbf{x}_\URUT|\bm{\Theta}_\UR}&=\mathbf{C}_{\mathbf{x}_\URUT\mathbf{y}_\UR|\bm{\Theta}_\UR}^H=\sqrt{\rho_\UT}\esp\Bigl\{\mathbf{A}_\UR\bm{\Lambda}_\UR^{1/2}|\bm{\Theta}_\UR\Bigr\}\label{eq:yrxur}.
\end{align}
By plugging (\ref{eq:cxurxur})-(\ref{eq:yrxur}) into \eqref{eq:xuruthat}, and then replacing $h(\mathbf{x}_\URUT|\bm{\Theta}_\UR)$ and \eqref{eq:hxy} into \eqref{eq:ixyr}, the SE at R can be computed as in \eqref{eq:SEr} by employing the matrix inversion lemma. Analogously, to obtain the SE at the CPU, the mutual information must be computed between $\mathbf{x}_\URUT$ and ${\mathbf{z}}_\MN$ in \eqref{eq:tildexrt}. 
\balance
\bibliographystyle{IEEEtran}  
\bibliography{references}

\end{document}